\documentclass[a4paper,UKenglish]{lipics-v2016}
\pdfoutput=1
\usepackage{amsmath}
\usepackage{amsthm}
\usepackage{amsfonts}
\usepackage{amssymb}
\usepackage{graphicx} 
\usepackage{microtype}%if unwanted, comment out or use option "draft"
\usepackage{todonotes}
\usepackage{enumitem}
\usepackage{framed}
\usepackage[ruled,boxed,linesnumbered,vlined,noend]{algorithm2e}
\usepackage{csquotes}

\usepackage{lmodern}
\usepackage[utf8]{inputenc}
\usepackage[T1]{fontenc}
\usepackage{microtype}

\usepackage{graphicx}

\usepackage{float}
\usepackage{makecell}

\newcommand{\opt}{\mathit{OPT}}
\newcommand{\stp}{\mathit{STP}}

\newcommand{\borchersdu}{\textsc{RestrictedST}}
\newcommand{\connect}{\textsc{Connect}}
\newcommand{\guessold}{\textsc{BuildST}}
\newcommand{\p}{\ensuremath{\textsf{\upshape P}}}
\newcommand{\np}{\ensuremath{\textsf{\upshape NP}}}

\title{New algorithms for Steiner tree reoptimization}
\titlerunning{New algorithms for Steiner tree reoptimization} %optional, in case that the title is too long; the running title should fit into the top page column

\author{Davide Bilò}
\affil{Department of Humanities and Social Sciences, University of Sassari\\via Roma 151, 07100 Sassari (SS), Italy.\\E-mail: \texttt{davidebilo@uniss.it}}
\authorrunning{D. Bilò} %mandatory. First: Use abbreviated first/middle names. Second (only in severe cases): Use first author plus 'et. al.'

\Copyright{Davide Bilò}%mandatory, please use full first names. LIPIcs license is "CC-BY";  http://creativecommons.org/licenses/by/3.0/

% mandatory: Please choose ACM 1998 classifications from http://www.acm.org/about/class/ccs98-html . E.g., cite as "F.1.1 Models of Computation". 
% mandatory: Please provide 1-5 keywords
\begin{document}

\maketitle

\begin{abstract}
{\em Reoptimization} is a setting in which we are given an (near) optimal solution of a  problem instance and a local modification that slightly changes the instance. The main goal is that of finding an (near) optimal solution of the modified instance.

We investigate one of the most studied scenarios in reoptimization  known as {\em Steiner tree reoptimization}. Steiner tree reoptimization is a collection of strongly $\np$-hard optimization problems that are defined on top of the classical Steiner tree problem and for which several constant-factor approximation algorithms have been designed in the last decade. 
In this paper we improve upon all these results by developing a novel technique that allows us to design {\em polynomial-time approximation schemes}. Remarkably, prior to this paper, no approximation algorithm better than recomputing a solution from scratch was known for the elusive scenario in which the cost of a single edge decreases. Our results are best possible since none of the problems addressed in this paper admits a fully polynomial-time approximation scheme, unless $\p=\np$.
%Remarkably, our algorithms can be combined together to compute near optimal solutions even for a sequence (of any type) of local modifications, where the length of the sequence is bounded by a poly-logarithm in the size of the old problem instance.
\end{abstract}

%\newpage

\section{Introduction}
{\em Reoptimization} is a setting in which we are given an instance $I$ of an optimization problem together with an (near) optimal solution $S$ for $I$ and a local modification that, once applied to $I$, generates a new instance $I'$ which slightly differs from $I$. The main goal is that of finding an (near) optimal solution of $I'$. Since reoptimization problems defined on top of $\np$-hard problems are usually $\np$-hard, the idea beyond reoptimization is to compute a good approximate solution of $I'$ by exploiting the structural properties of $S$. This approach leads to the design of reoptimization algorithms that, when compared to classical algorithms that compute feasible solutions from scratch, are less expensive in terms of running time or output solutions which guarantee a better approximation ratio. Thus, reoptimization finds applications in many real settings (like scheduling problems or network design problems) where prior knowledge is often at our disposal and a problem instance can arise from a local modification of a previous problem instance (for example, a scheduled job is canceled or some links of the network fail).

 The term reoptimization was mentioned for the first time in the paper of Sch{\"{a}}ffter~\cite{DBLP:journals/dam/Schaffter97}, where the author addressed an $\np$-hard scheduling problem with forbidden sets in the scenario of adding/removing a forbidden set. Since then, reoptimization has been successfully applied to other $\np$-hard problems including: the Steiner tree problem~\cite{DBLP:conf/swat/BiloBHKMWZ08, DBLP:conf/mfcs/BiloZ12, DBLP:journals/jda/BockenhauerFHMSS12, DBLP:journals/tcs/BockenhauerHKMR09, DBLP:conf/sofsem/BockenhauerHMW08, DBLP:journals/aor/EscoffierMP09, DBLP:conf/fsttcs/GoyalM15, DBLP:journals/endm/ZychB11}, the traveling salesman problem and some of its variants~\cite{DBLP:journals/networks/ArchettiBS03, DBLP:journals/jda/AusielloEMP09,DBLP:conf/lata/BergH09, DBLP:journals/aor/BockenhauerFHKKPW07, DBLP:conf/sofsem/BockenhauerHMW08, DBLP:conf/birthday/BockenhauerHS11, DBLP:journals/fuin/BockenhauerHS11, DBLP:journals/tcs/BockenhauerKK09, DBLP:journals/jda/BockenhauerK10, DBLP:journals/ipl/Monnot15}, scheduling problems~\cite{DBLP:journals/algorithmica/BenderFFFG15, DBLP:journals/tcs/BoriaC14}, the minimum latency problem~\cite{DBLP:conf/cocoon/Dai17}, the minimum spanning tree problem~\cite{DBLP:journals/jda/BoriaP10}, the rural postman problem~\cite{DBLP:journals/cor/ArchettiGS13}, many covering problems~\cite{DBLP:conf/waoa/BiloWZ08}, the shortest superstring problem~\cite{DBLP:journals/algorithmica/BiloBKKMSZ11}, the knapsack problem~\cite{DBLP:journals/dam/ArchettiBS10}, the maximum-weight induced hereditary subgraph problems~\cite{DBLP:journals/tcs/BoriaMP13}, and the maximum $P_k$ subgraph problem~\cite{DBLP:journals/dmaa/BoriaMP13}. Some overviews on reoptimization can be found in~\cite{doi:10.1142/9781848162778_0004, DBLP:journals/rairo/BoriaP11, AnnaZychPhDthesis}.\footnote{In this paper we focus only on  reoptimization problems that are $\np$-hard.}

\subsection{Our results}

In this paper we address {\em Steiner tree reoptimization}. Steiner tree reoptimization is a collection of optimization problems that are built on top of the famous {\em Steiner tree problem} ($\stp$ for short), a problem that given an edge-weighted graph as input, where each vertex can be either terminal or Steiner, asks to find a minimum-cost tree spanning all the terminals. More precisely, we study the four local modifications in which the cost of an edge changes (it can either increase or decrease) or a vertex changes its status (from Steiner to terminal and viceversa).\footnote{We observe that the insertion of an edge can be modeled by decreasing the cost of the edge from $+\infty$ to its real value. Analogously, the deletion of an edge can be modeled by increasing the cost of the edge to $+\infty$.} All the problems addressed are strongly $\np$-hard~\cite{DBLP:journals/tcs/BockenhauerHKMR09, DBLP:conf/sofsem/BockenhauerHMW08}, which implies that none of them admits a fully polynomial-time approximation scheme, unless $\p = \np$.

In this paper we improve upon all the constant-factor approximation algorithms that have been developed in the last 10 years by designing {\em polynomial-time approximation schemes}. More precisely, if $S$ is a $\rho$-approximate solution of $I$, then all our algorithms compute a $(\rho+\epsilon)$-approximate solution of $I'$. For the scenarios in which the cost of an edge decreases or a terminal vertex becomes Steiner, our polynomial-time approximation schemes hold under the assumption that $\rho=1$, i.e., the provided solution is optimal. We observe that this assumption is somehow necessary since Goyal and M{\"{o}}mke (see~\cite{DBLP:conf/fsttcs/GoyalM15}) proved that in both scenarios, unless $\rho=1$, any $\sigma$-approximation algorithm for the reoptimization problem would be also a $\sigma$-approximation algorithm for the $\stp$! Remarkably, prior to this paper, no approximation algorithm better than recomputing a solution from scratch was known for the elusive scenario in which the cost of a single edge decreases.

The state of the art of Steiner tree reoptimization is summarized in Table~\ref{table}. We observe that the only problem that remains open is the design of a reoptimization algorithm for the scenario in which a vertex, that can be either terminal or Steiner, is added to the graph. In fact, as proved by Goyal and M{\"{o}}mke~\cite{DBLP:conf/fsttcs/GoyalM15}, the scenario in which a vertex is removed from the graph is as hard to approximate as the $\stp$ (think of the removal of a vertex that is connected to each terminal by an edge of cost 0).

%In the last section of the paper, we also show how to combine all our algorithms to obtain a polynomial-time approximation scheme for the case in which we have a sequence of  (any of the above four types of) local modifications, where the length of the sequence is upper bounded by any poly-logarithm in $n$. We observe that this result also captures the scenario in which vertices of poly-logarithmic bounded degree are added to/removed from the graph as, for instance, the insertion of a terminal vertex $v$ can be modeled by a sequence of edge insertions, followed by the local modification which changes the status of $v$ from Steiner to terminal. To the best of our knowledge, prior to this work, no approximation algorithm better than recomputing a solution from scratch was known for the vertex addition/removal scenario, even for the case of 2-degree vertices!

\begin{table}[ht]
\caption{The state of the art of the approximability of Steiner tree reoptimization before and after our paper. The value $\sigma$ is the best approximation ratio we can achieve for $\stp$ (actually $\sigma =\ln 4 + \epsilon\approx 1.387$~\cite{DBLP:journals/jacm/ByrkaGRS13}). If we substitute $\sigma =\ln 4 + \epsilon$, we get $\frac{10\sigma-7}{7\sigma-4}\approx 1.204, \frac{7\sigma-4}{4\sigma-1}\approx 1.256$, and $\frac{5\sigma-3}{3\sigma-1}\approx 1.246$. The bound of $\frac{5\sigma-3}{3\sigma-1}$ for the scenario in which the cost of an edge decreases holds when the modified edge can affect only its cost in the shortest-path metric closure. To increase readability, all the bounds are provided for the case in which $\rho=1$. Local modifications marked with an asterisk are those for which the condition $\rho=1$ is necessary as otherwise the reoptimization problem would be as hard to approximate as the $\stp$. Finally, the question mark means that the problem is still open.} % title of Table
\centering % used for centering table
\begin{tabular}{|l |c |c|} % centered columns (4 columns)
\hline %inserts double horizontal lines
{\bf Local modification} & {\bf before our paper} & {\bf after our paper}\\[1ex] % inserts table
%heading
\hline % inserts single horizontal line
a terminal becomes a Steiner vertex* & $\frac{10\sigma-7}{7\sigma-4}+\epsilon$~\cite{DBLP:conf/fsttcs/GoyalM15} & $1+\epsilon$  \\
\hline % inserting body of the table
a Steiner vertex becomes terminal  & $\frac{10\sigma-7}{7\sigma-4}+\epsilon$~\cite{DBLP:conf/fsttcs/GoyalM15} & $1+\epsilon$  \\
\hline
an edge cost increases & $\frac{7\sigma-4}{4\sigma-1}+\epsilon$~\cite{DBLP:conf/fsttcs/GoyalM15} & $1+\epsilon$  \\
\hline
an edge cost decreases* & $\frac{5\sigma-3}{3\sigma-1}$~\cite{DBLP:conf/swat/BiloBHKMWZ08} & $1+\epsilon$  \\
% [1ex] adds vertical space
\hline %inserts single line
a vertex is added to the graph*	& ?	& ? \\
\hline
a vertex is deleted from the graph & as hard as the $\stp$~\cite{DBLP:conf/fsttcs/GoyalM15} & as hard as the $\stp$~\cite{DBLP:conf/fsttcs/GoyalM15}\\
\hline
\end{tabular}
\label{table} % is used to refer this table in the text
\end{table}

\subsection{Used techniques}
All the polynomial-time approximation schemes proposed in this paper make a clever use of the following three algorithms:
\begin{itemize}
\item the algorithm $\connect$ that computes a Steiner tree in polynomial time by augmenting a Steiner forest having a constant number of trees with a minimum-cost set of edges. This algorithm has been introduced by B\"ockenhauer et al. in~\cite{DBLP:journals/tcs/BockenhauerHKMR09} and has been subsequently used in several other papers on Steiner tree reoptimization~\cite{DBLP:conf/mfcs/BiloZ12, DBLP:journals/jda/BockenhauerFHMSS12, DBLP:conf/fsttcs/GoyalM15, DBLP:journals/endm/ZychB11};
\item the algorithmic proof of Borchers and Du~\cite{DBLP:journals/siamcomp/BorchersD97} that, for every $\xi > 0$, converts a Steiner tree $S$ into a {\em $f(\xi)$-restricted Steiner tree} $S_{\xi}$, for some function $f(\xi)$ that depends on $\xi$ only, whose cost is a $(1+\xi)$-factor away from the cost of $S$. This algorithm has been used for the first time in Steiner tree reoptimization by Goyal and M{\"{o}}mke~\cite{DBLP:conf/fsttcs/GoyalM15};
\item a novel algorithm, that we call $\guessold$, that takes a $f(\xi)$-restricted Steiner forest $S_{\xi}$ of $S$ with $q$ trees, generated with the algorithm of Borchers and Du, and a positive integer $h$ as inputs and computes a minimum-cost Steiner tree $S'$ w.r.t. the set of all the feasible solutions that can be obtained by swapping up to $h$ {\em full components} of $S_{\xi}$ with a minimum-cost set of edges (that is computed using the algorithm $\connect$). The running time of this algorithm is polynomial when all the three parameters $f(\xi),q$, and $h$ are bounded by a constant.
\end{itemize}

We prove that the approximation ratio of the solution $S'$ returned by the algorithm $\guessold$ is $\rho+\epsilon$ (i) by showing  that the cost of $S'$ is at most the cost of other feasible solutions $S_1, \dots, S_\ell$ and (ii) by proving useful upper bounds to the cost of each $S_i$’s.
Intuitively, each $S_i$ is obtained by swapping up to $h$ suitable full components, say $H_i$, of a suitable $f(\xi)$-restricted version of a forest in $S$, say $S_{\xi}$, with a minimal set of full components, say $H'_i$, of a $g(\xi)$-restricted version of a fixed optimal solution $\opt'$. This implies that we can easily bound the cost of each $S_i$ w.r.t. the cost of $S_{\xi}$, the cost of $H_i$, and the cost of $H'_i$. Unfortunately, none of these bounds can be used by itself to prove the claim. We address this issue by carefully defining $H_{i+1}$ as a function of both $\opt'$ and $H'_i$, and $H'_{i}$ as a function of both $S_{\xi}$ and $H_i$. This trick allows us to derive better upper bounds: we prove that the cost of each $S_i$ is at most $\rho g(\xi)^2$ times the cost of $\opt'$ plus the $i$-th term of a telescoping sum. As each term of the telescoping sum can be upper bounded by $g(\xi)$ times the cost of $\opt'$, the bound of $(\rho+\epsilon)$ on the approximation ratio of $S'$ then follows by averaging the costs of the $S_i$'s and because of the choices of the parameters. 

%The bound on the approximation ratio of all our algorithms is obtained by showing that the cost of the computed Steiner tree is at most the cost of each of the following feasible solutions $S_1,\dots,S_\ell$, where each $S_i$ is obtained by swapping $h$ full components of $F$ with a minimum-cost set of full components of an optimal solution. Then, we show that (a) the cost of each $S_i$ is upper bounded by $\rho(1+\xi)^2$ times the cost of a new optimal solution plus the $i$-th term of a telescoping sum and (b) the $\ell$-th term of the telescoping sum is also bounded by $(1+\xi)$ times the cost of an optimal solution. 

 The paper is organized as follows: in Section~\ref{definition_notation} we provide the basic definitions and some preliminaries; in Section~\ref{general_tools} we describe the three main tools that are used in all our algorithms; in Sections~\ref{R_plus},~\ref{E_plus},~\ref{R_minus}, and~\ref{E_minus} we describe and analyze the algorithms for all the four local modifications addressed in this paper.
 
%the case in which a Steiner vertex becomes terminal and the cost of an edge increases, respectively. Due to the lack of space, the local modifications in which a terminal becomes a Steiner vertex and the cost of an edge decreases can be found in the Appendix~\ref{R_minus} and Appendix~\ref{E_minus}, respectively.

\section{Basic definitions and preliminaries}\label{definition_notation}

The Steiner tree problem ($\stp$ for short) is defined as follows: 
\begin{itemize}
\item the input is a triple $I=\langle G, c, R\rangle$ where $G=(V(G),E(G))$ is a connected undirected graph with $|V(G)|=n$ vertices, $c$ is a function that associates a real value $c(e)\geq 0$ to each edge $e \in E(G)$, and $R\subseteq V(G)$ is a set of {\em terminal} vertices;
\item the problem asks to compute a minimum-cost Steiner tree of $I$, i.e., a tree that spans $R$ and that minimizes the overall sum of its edge costs.
\end{itemize}
The vertices in $V(G)\setminus R$ are also called {\em Steiner} vertices. With a slight abuse of notation, for any subgraph $H$ of $G$, we denote by $c(H):=\sum_{e \in E(H)}c(e)$ the cost of $H$. %It is well known that $\stp$ is $\apx$-complete~\cite{DBLP:journals/ipl/BernP89} as well as approximable within $\ln 4+\epsilon\approx 1.387$~\cite{DBLP:journals/jacm/ByrkaGRS13}.

In {\em Steiner tree reoptimization} we are given a triple $\langle I, S, I'\rangle$ where $I$ is an $\stp$ instance, $S$ is a $\rho$-approximate Steiner tree of $I$, and $I'$ is another $\stp$ instance which slightly differs from $I$. The problem asks to find a $\rho$-approximate Steiner tree of $I'$.

For a forest $F$ of $G$ (i.e., with $E(F)\subseteq E(G)$) and a set of edges $E'\subseteq E(G)$, we denote by $F+E'$ a (fixed) \underline{forest} yielded by the addition of all the edges in $E'\setminus E(F)$ to $F$, i.e., we scan all the edges of $E'$ one by one in any (fixed) order and we add the currently scanned edge to the forest only if the resulting graph remains a forest. Analogously, we denote by $F-E'$ the forest yielded by the removal of all the edges in $E'\cap E(F)$ from $F$. We use the shortcuts $F+e$ and $F-e$ to denote $F+\{e\}$ and $F-\{e\}$, respectively. Furthermore, if $F'$ is another forest of $G$ we also use the shortcuts $F+F'$ and $F-F'$ to denote $F+E(F')$ and $F-E(F')$, respectively. 
For a set of vertices $U \subseteq V(G)$, we define $F+U:=(V(F)\cup U, E(F))$ and use the shortcut $F+v$ to denote $F+\{v\}$. %Finally, we denote by $F-U$ the forest obtained by removing from $F$ all the vertices $v \in U$, together with all their incident edges, and use the shortcut $F-v$ to denote $F-\{v\}$.

W.l.o.g., in this paper we tacitly assume that all the leaves of a Steiner tree are terminals: if a Steiner tree $S$ contains a leaf which is not a terminal, then we can always remove such a leaf from $S$ to obtain another Steiner tree whose cost is upper bounded by the cost of $S$. Analogously, we tacitly assume that all the leaves of a Steiner forest (i.e., a forest spanning all the terminals) are terminals.
A {\em full component} of a Steiner forest $F$ -- so a Steiner tree as well -- is a maximal (w.r.t. edge insertion) subtree of $F$ whose leaves are all terminals and whose internal vertices are all Steiner vertices. We observe that a Steiner forest can always be decomposed into full components. A {\em $k$-restricted} Steiner forest is a Steiner forest where each of its full components has at most $k$ terminals.

In the rest of the paper, for any superscript $s$ (the empty string is included), we denote by $I^s$ an $\stp$ instance (we tacitly assume that $I^s=\langle G^s,c^s,R^s\rangle$ as well as that $G^s$ has $n$ vertices), by $\opt^s$ a fixed optimal solution of $I^s$, by $d^s$ the shortest-path metric induced by $c^s$ in $G^s$ (i.e, $d^s(u,v)$ is equal to the cost of a shortest path between $u$ and $v$ in $G^s$ w.r.t. cost function $c^s$), by $K^s$ the complete graph on $V(G^s)$ with edge costs $d^s$, and by $K^s_n$ the complete graph obtained by adding $n-1$ copies of $K^s$ to $K^s$ and where the cost of an edge between a vertex and any of its copies is equal to 0. With a little abuse of notation, we use $d^s$ to denote the cost function of $K^s_n$. Moreover, we say that a subgraph $H$ of $K^s_n$ {\em spans} a vertex $v \in V(G^s)$ if $H$ spans -- according to the standard terminology used in graph theory -- any of the $n$ copies of $v$. Finally, with a little abuse of notation, we denote by $I_n^s$ the $\stp$ instance $I_n^s=\langle K_n^s, d^s, R^s\rangle$.

Let $I$ be an $\stp$ instance. We observe that any forest $F$ of $G$ is also a forest of $K$ such that $d(F)\leq c(F)$. Moreover, any forest $\bar F$ of $K$ can be transformed in polynomial time into a forest $F$ of $G$ spanning $V(\bar F)$ and such that $c(F)\leq d(\bar F)$. In fact, $F$ can be build from an empty graph on $V(G)$ by iteratively adding -- according to our definition of graph addition --  a shortest path between $u$ and $v$ in $G$, for every edge $(u,v) \in E(\bar F)$.

We also observe that any forest $F$ of $K$ can be viewed as a forest of $K_n$. Conversely, any forest $F$ of $K_n$ can be transformed in polynomial time into a forest of $K$ spanning $V(F)$ and having the same cost of $F$: it suffices to identify each of the 0-cost edges of $F$ between any vertex and any of its copies.
Therefore, we have polynomial-time algorithms that convert any forest $F$ of any graph in $\{G, K, K_n\}$ into a forest $F'$ of any graph in $\{G,K,K_n\}$ such that $F'$ always spans $V(F)$ and the cost of $F'$ (according to the cost function associated with the corresponding graph) is less than or equal to the cost of $F$ (according to the cost function associated with the corresponding graph). All these polynomial-time algorithms will be implicitly used in the rest of the paper; for example, we can say that a Steiner tree of $I_n$ is also a Steiner tree of $I$ or that a forest of $I_n$ that spans $v$ is also a forest of $I$ that spans $v$. However, we observe that some of these polynomial-time algorithms do no preserve useful structural properties; for example, when transforming a $k$-restricted Steiner tree of $I_n$ into a Steiner tree of $I$ we may lose the property that the resulting tree is $k$-restricted. Therefore, whenever we refer to a structural property of a forest, we  assume that such a property holds only w.r.t. the underlying graph the forest belongs to.

\section{General tools}\label{general_tools}

In this section we describe the three main tools that are used by all our algorithms. 

The first tool is the algorithm $\connect$ that has been introduced by B\"ockenhauer et al. in~\cite{DBLP:journals/tcs/BockenhauerHKMR09}. This algorithm takes an $\stp$ instance $I$ together with a Steiner forest $F$ of $I$ as inputs and computes a minimum-cost set of edges of $G$ whose addition to $F$ yields a Steiner tree of $I$. The algorithm $\connect$ reduces the $\stp$ instance to a smaller $\stp$ instance, where each terminal vertex corresponds to a tree of $F$, and then uses the famous Dreyfus-Wagner algorithm (see~\cite{DreyfusWagner1971}) to find an optimal solution of the reduced instance.\footnote{We refer to Hougardy~\cite{DBLP:journals/mpc/HougardySV17} for further exact algorithms for the $\stp$.} If $F$ contains $q$ trees, the running time of the algorithm $\connect$ is $O^*(3^q)$;\footnote{The $O^*$ notation hides polynomial factors.} this implies that the algorithm has polynomial running time when $q$ is constant. The call of the algorithm $\connect$ with input parameters $I$ and $F$ is denoted by $\connect(I,F)$.

The second tool is the algorithmic proof of Borchers and Du on $k$-restricted Steiner trees~\cite{DBLP:journals/siamcomp/BorchersD97} that has already been used in Steiner tree reoptimization in the recent paper of Goyal and M{\"{o}}mke~\cite{DBLP:conf/fsttcs/GoyalM15}. In their paper, Borchers and Du proved that, for every $\stp$ instance $I$ and every $\xi > 0$, there exists a $k$-restricted Steiner tree $S_{\xi}$ of $I_n$, with $k=2^{\lceil {1/\xi}\rceil}$, such that $d(S_{\xi}) \leq (1+\xi)c(\opt)$. As the following theorem shows, the algorithmic proof of Borchers and Du immediately extends to any (not necessarily optimal!) Steiner tree of $I_n$.

\begin{theorem}[\cite{DBLP:journals/siamcomp/BorchersD97}]\label{thm:BorchersDu}
Let $I$ be an $\stp$ instance, $S$ a Steiner tree of $I$, and $\xi > 0$ a real value. There is a polynomial time algorithm that computes a $k$-restricted Steiner tree $S_{\xi}$ of $I_n$, with $k=2^{\lceil {1/\xi}\rceil}$, such that $d(S_{\xi}) \leq (1+\xi) c(S)$.
\end{theorem}
The call of Borchers and Du algorithm with input parameters $I$, $S$, and $\xi$, is denoted by $\borchersdu(I,S,\xi)$. By construction, the algorithm of Borchers and Du also guarantees the following properties: 
\begin{itemize}
\item[(a)] if $v$ is a Steiner vertex that has degree at least 3 in $S$, then $S_{\xi}$ spans $v$;
\item[(b)] if the degree of a terminal $t$ in $S$ is 2, then the degree of $t$ in $S_{\xi}$ is at most 4.
\end{itemize}
As shown in the following corollary, these two properties are useful if we want that a specific vertex of $S$ would also be a vertex of the $k$-restricted Steiner tree $S_{\xi}$ of $I_n$.
\begin{corollary}\label{cor:BorchersDu}
Let $I$ be an $\stp$ instance, $S$ a Steiner tree of $I$, $v$ a vertex of $S$, and $\xi > 0$ a real value. There is a polynomial time algorithm that computes a $k$-restricted Steiner tree $S_{\xi}$ of $I_n$, with $k=2^{2+\lceil {1/\xi}\rceil}$, such that $d(S_{\xi}) \leq (1+\xi) c(S)$ and $v \in V(S_{\xi})$.
\end{corollary}
\begin{proof}
The cases in which $v \in R$ has already been proved in Theorem~\ref{thm:BorchersDu}. Moreover, the claim trivially holds by property (a) when $v$ is a Steiner vertex of degree greater than or equal to 3 in $S$. Therefore we assume that $v$ is a Steiner vertex of degree 2 in $S$. Let $I' = \langle G,c,R\cup\{v\}\rangle$. In this case, the call of $\borchersdu\big(I',S,\xi\big)$ outputs a $k'$-restricted Steiner tree $S_{\xi}$ of $I'_n$, with $k'=2^{\lceil {1/\xi}\rceil}$, such that $d(S_{\xi}) \leq (1+\xi) c(S)$. Since $v$ is a terminal in $I'$, by property (b) the degree of $v$ in $S_{\xi}$ is at most 4; in other words, $v$ is contained in at most 4 full components of $S_{\xi}$. Therefore, $S_{\xi}$ is a $k$-restricted Steiner tree of $I$, with $k=4 k'=2^{2+\lceil {1/\xi}\rceil}$.
\end{proof}
The call of Borchers and Du algorithm with parameters $I$, $S$, $\xi$, and $v$ is denoted by $\borchersdu(I,S,\xi,v)$.

The third tool, which is the novel idea of this paper, is the algorithm $\guessold$ (see Algorithm~\ref{alg:guessold} for the pseudocode). The algorithm $\guessold$ takes as inputs an $\stp$ instance $I$, a Steiner forest $F$ of $I_n$, and a positive integer value $h$. The algorithm computes a Steiner tree of $I$ by selecting, among several feasible solutions, the one of minimum cost. More precisely, the algorithm computes a feasible solution for each forest $\bar F$ that is obtained by removing from $F$ up to $h$ of its full components. The Steiner tree of $I$ associated with $\bar F$ is obtained by the call of $\connect\big(I_n,\bar F\big)$. If $F$ contains a number of full components which is strictly less than $h$, then the algorithm computes a minimum-cost Steiner tree of $I$ from scratch. The call of the algorithm $\guessold$ with input parameters $I$, $F$, and $h$ is denoted by $\guessold(I,F, h)$. The proof of the following lemma is immediate.
\begin{algorithm}[t!]
\footnotesize
\caption{The pseudocode of the algorithm $\guessold(I,F,h)$.}
	\label{alg:guessold}

	\DontPrintSemicolon
	{\bf if} $F$ contains less than $h$ full components {\bf then return} $\connect(I,(R,\emptyset))$;\;
	$S' \gets \perp$;\;
	\For{every set $H$ of up to $h$ full components of $F$}
	{
		$\bar F\gets F-H$;\;
		$\bar S \gets \connect\big(I_n,\bar F\big)$;\;
		{\bf if} $S' = \perp$ or $c(\bar S)< c(S')$ {\bf then} $S'\gets \bar S$;\;
	}
	\Return $S'$;
\end{algorithm}

\begin{lemma}\label{lm:merging_opts}
Let $I$ be an $\stp$ instance, $F$ a Steiner forest of $I_n$, $\bar F$ a Steiner forest of $I_n$ that is obtained from $F$ by removing up to $h$ of its full components, and $H$ a minimum-cost subgraph of $K_n$ such that $\bar F+H$ is a Steiner tree of $I$. Then the cost of the Steiner tree of $I$ returned by the call of $\guessold(I,F,h)$ is at most $c(\bar F)+c(H)$.
\end{lemma}

\noindent Moreover, the following lemmas hold.

\begin{lemma}\label{lm:full_connected_components}
Let $I$ be an $\stp$ instance and $F$ a forest of $G$ spanning $R$. If $F$ contains $q$ trees, then every minimal (w.r.t. the deletion of entire full components) subgraph $H$ of any graph in $\{G,K,K_n\}$ such that $F+H$ is a Steiner tree of $I$ contains at most $q-1$ full components.
\end{lemma}
\begin{proof}
Since $H$ is minimal, each full component of $H$ contains at least 2 vertices each of which respectively belongs to one of two distinct trees of $F$. Therefore, if we add the full components of $H$ to $F$ one after the other, the number of connected components of the resulting graph strictly decreases at each iteration because of the minimality of $H$. Hence, $H$ contains at most $q-1$ full components.
\end{proof}

\begin{lemma}\label{lm:guessold}
Let $I$ be an $\stp$ instance and $F$ a Steiner forest of $I_n$ of $q$ trees. If each tree of $F$ is $k$-restricted, then the call of $\guessold(I,F,h)$ outputs a Steiner tree of $I$ in time $O^*(|R|^h 3^{q+hk})$.
\end{lemma}
\begin{proof}
If $F$ contains at most $h$ full components, then the number of terminals is at most $q+hk$ and the call of $\connect(I,(R,\emptyset))$ outputs a Steiner tree of $I$ in time $O^*(3^{q+hk})$. Therefore, we assume that $F$ contains a number of full components that is greater than or equal to $h$.
In this case, the algorithm $\guessold(I,F,h)$ computes a feasible solution for each forest that is obtained by removing up to $h$ full components of $F$ from itself. As the number of full components of $F$ is at most $|R|$, the overall number of feasible solutions evaluated by the algorithm is $O(|R|^{h+1})$. Let $\bar F$ be any forest  that is obtained from $F$ by removing $h$ of its full components. Since $F$ is $k$-restricted, $\bar F$ contains at most $q+hk$ trees. Therefore the call of $\connect(I_n, \bar F)$ requires $O^*(3^{q+hk})$ time to output a solution. Hence, the overall time needed by the call of $\guessold(I,F,h)$ to output a solution is $O^*(|R|^h3^{q+hk})$.
\end{proof}

For the rest of the paper, we denote by $\langle I, S, I'\rangle$ an instance of the Steiner tree reoptimization problem. Furthermore, we also assume that $\rho \leq 2$ as well as $\epsilon \leq 1$, as otherwise we can use classical time-efficient algorithms to compute a 2-approximate solution of $I'$.

\section{A Steiner vertex becomes a terminal}\label{R_plus}

In this section we consider the local modification in which a Steiner vertex $t \in V(G)\setminus R$ becomes a terminal. Clearly, $G'=G$, $c'=c$, $d'=d$, whereas $R'=R\cup \{t\}$. Therefore, for the sake of readability, we will drop the superscripts from $G'$, $c'$, and $d'$. 

Let $\xi= \epsilon/10$ and $h =2^{2\lceil 2/\epsilon \rceil \lceil 1/\xi\rceil}$.
The algorithm, whose pseudocode is reported in Algorithm~\ref{alg:R_plus}, computes a Steiner tree of $I'$ first by calling $\borchersdu(I,S,\xi)$ to obtain a $k$-restricted Steiner tree $S_{\xi}$ of $I_n$, with $k=2^{\lceil {1/\xi}\rceil}$, and then by calling $\guessold(I',S_{\xi}+t,h)$. Intuitively, the algorithm {\em guesses} the full components of $S_{\xi}$ to be replaced with the corresponding full components of a suitable $k'$-restricted version of a new (fixed) optimal solution, for suitable values of $k'$.

\begin{algorithm}[h]
\footnotesize
\caption{A Steiner vertex $t \in V(G)\setminus R$ becomes a terminal.}
	\label{alg:R_plus}

	\DontPrintSemicolon
	$\xi \gets \epsilon/10$;\;
	$h \gets 2^{2\lceil 2/\epsilon \rceil \lceil 1/\xi\rceil}$;\;
	$S_{\xi} \gets \borchersdu(I,S,\xi)$;\;
	$S' \gets \guessold(I',S_{\xi}+t,h)$;\;
	\Return $S'$;\;
\end{algorithm}

\begin{theorem}\label{thm:R_plus}
Let $\langle I, S, I' \rangle$ be an instance of Steiner tree reoptimization, where $S$ is a $\rho$-approximate solution of $I$ and $I'$ is obtained from $I$ by changing the status of a single vertex from Steiner to terminal. Then Algorithm~\ref{alg:R_plus} computes a $(\rho+\epsilon)$-approximate Steiner tree of $I'$ in polynomial time.
\end{theorem}
\begin{proof}
Theorem~\ref{thm:BorchersDu} implies that the computation of $S_{\xi}$ requires polynomial time. Since $S_{\xi}$ is a Steiner tree of $I_n$, $S_{\xi}+t$ is a Steiner forest of $I'_n$ of at most two trees (observe that $S_{\xi}$ may already contain $t$). Therefore, using Lemma~\ref{lm:guessold} and the fact that $h$ is a constant value, the call of $\guessold(I',S_{\xi}+t,h)$ outputs a solution $S'$ in polynomial time. Hence, the overall time complexity of Algorithm~\ref{alg:R_plus} is polynomial. 

In the following we show that Algorithm~\ref{alg:R_plus} returns a $(\rho+\epsilon)$-approximate solution. Let $S'_{\xi}$ denote the Steiner tree of $I'_n$ that is returned by the call of $\borchersdu(I',\opt',\xi)$.
Let $H'_0$ be a full component of $S'_{\xi}$ that spans $t$ (ties are broken arbitrarily). For every $i=1,\dots,\ell$, with $\ell= \lceil 2/\epsilon \rceil$, we define:
\begin{itemize}
\item $H_i$ as the forest consisting of a minimal set of full components of $S_{\xi}$ whose addition to $S'_{\xi}-H'_{i-1}$ yields a Steiner tree of $I_n$ (see Figure~\ref{fig:steiner_to_terminal} (b) and (d));
\item $H'_i$ as the forest consisting of a minimal set of full components of $S'_{\xi}$ whose addition to $S_{\xi}-H_i$ yields a Steiner tree of $I'_n$ (see Figure~\ref{fig:steiner_to_terminal} (c)).
\end{itemize}

For the rest of the proof, we denote by $S_i$ the Steiner tree of $I'$ yielded by the addition of $H'_i$ to $S_{\xi}-H_i$. 

\begin{figure}[t]
	\centering
	\includegraphics[scale=0.6]{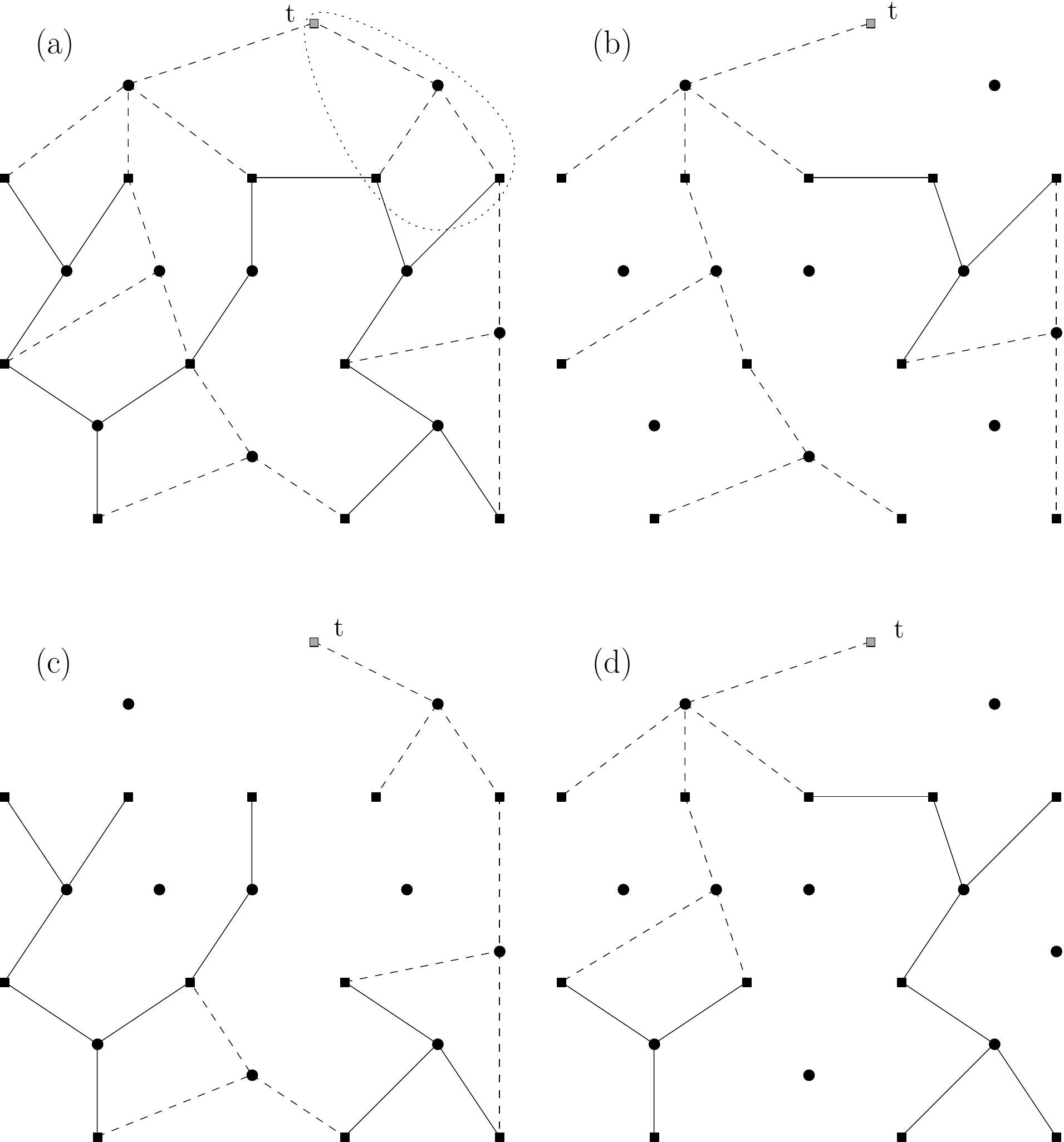}
	\caption{An example that shows how $H_i$ and $H'_i$ are defined. In (a) it is depicted a Steiner tree reoptimization instance where square vertices represent the terminals, while circle vertices represent the Steiner vertices. The Steiner vertex $t$ of $I$ that becomes a terminal in $I'$ is represented as a terminal vertex of grey color. Solid edges are the edges of $S_{\xi}$, while dashed edges are the edges of $S'_{\xi}$. The full component $H'_0$ of $S'_{\xi}$ is highlighted in (a). The forests $S'_{\xi}-H'_0$ and $H_1$ are shown in (b). The forests $S_{\xi}-H_1$ and $H'_1$ are shown in (c). Finally, the forests $S'_{\xi}-H'_1$ and $H_2$ are shown in (d).}
    \label{fig:steiner_to_terminal}
\end{figure}

Let $k=2^{\lceil {1/\xi}\rceil}$ and observe that both $S_{\xi}$ and $S'_{\xi}$ are $k$-restricted Steiner trees of $I_n$ and $I'_n$, respectively (see Theorem~\ref{thm:BorchersDu}). Therefore $H'_0$ spans at most $k$ terminals. Furthermore, by repeatedly using Lemma~\ref{lm:full_connected_components}, $H_i$ contains at most $k^{2i-1}$ full components and spans at most $k^{2i}$ terminals, while $H'_i$ contains at most $k^{2i}$ full components and spans at most $k^{2i+1}$ terminals. 
This implies that the number of the full components of $H_i$, with $i=1,\dots,\ell$, is at most
\begin{equation*}
k^{2i-1} \leq k^{2\ell}= 2^{2\lceil 2/\epsilon \rceil \lceil {1/\xi}\rceil} = h.
\end{equation*}
Therefore, by Lemma~\ref{lm:merging_opts}, the cost of the solution returned by the call of $\guessold(I',S_{\xi}+t,h)$ is at most $c(S_i)$. As a consequence
\begin{equation}\label{eq:R_plus_sixth}
c(S') \leq \frac{1}{\ell}\sum_{i=1}^\ell c(S_i).
\end{equation}

Now we prove an upper bound to the cost of each $S_i$.
Let $\Delta_{\opt}=(1+\xi)c(\opt')-c(\opt)$. As any Steiner tree of $I'$ is also a Steiner tree of $I$, $c(\opt')\geq c(\opt)$; therefore $\Delta_{\opt}\geq 0$. Since the addition of $H_{i}$ to $S'_{\xi}-H'_{i-1}$ yields a Steiner tree of $I$, the cost of this tree is lower bounded by the cost of $\opt$. As a consequence, using Theorem~\ref{thm:BorchersDu} in the second inequality that follows, $c(\opt)\leq d(S'_{\xi}) - d(H'_{i-1}) + d(H_{i})\leq (1+\xi)c(\opt') - d(H'_{i-1}) + d(H_{i})$, i.e., 
\begin{equation}\label{eq:R_plus_first}
d(H_{i}) \geq  d(H'_{i-1}) - \Delta_{\opt}. 
\end{equation}
Using Theorem~\ref{thm:BorchersDu} in the first inequality that follows we have that
\begin{equation}\label{eq:R_plus_second}
d(S_{\xi}) 	\leq (1+\xi)c(S) \leq \rho(1+\xi)c(\opt)\leq\rho(1+\xi)^2 c(\opt')-\Delta_{\opt}.
\end{equation}
Using both~(\ref{eq:R_plus_first}) and~(\ref{eq:R_plus_second}) in the second inequality that follows, we can upper bound the cost of $S_i$ with
\begin{equation}\label{eq:R_plus_fourth}
c(S_i) 	\leq d(S_{\xi})-d(H_i)+d(H'_i) \leq \rho (1+\xi)^2c(\opt')-d(H'_{i-1})+d(H'_i).
\end{equation}
Observe that the overall sum of the last two terms on the right-hand side of~(\ref{eq:R_plus_fourth}) for every $i\in\{1,\dots,\ell\}$ is a telescoping sum. As a consequence, if for every $i\in\{1,\dots,\ell\}$ we substitute the term $c(S_i)$ in~(\ref{eq:R_plus_sixth}) with the corresponding upper bound in~(\ref{eq:R_plus_fourth}), and use Theorem~\ref{thm:BorchersDu} to derive that $d(H'_\ell) \leq d(S'_{\xi}) \leq (1+\xi)c(\opt')$, we obtain
\begin{equation*}
c(S')\leq \frac{1}{\ell}\sum_{i=1}^\ell c(S_i) \leq  \rho(1+\xi)^2 c(\opt')+\frac{1+\xi}{\ell} c(\opt')\leq (\rho+\epsilon)c(\opt'),
\end{equation*}
where last inequality holds by the choice of $\xi$ and $\ell$, and because $\rho \leq 2$ and $\epsilon \leq 1$. This completes the proof.
\end{proof}

\section{The cost of an edge increases}\label{E_plus}

In this section we consider the local modification in which the cost of an edge $e=(u,v) \in E(G)$ increases by $\Delta > 0$. More formally, we have that $G'=G$, $R'=R$, whereas, for every $e' \in E(G)$,
\[
c'(e')=
	\begin{cases}
		c(e) + \Delta	&	\text{if $e'=e$;}\\
		c(e')			&	\text{otherwise.}
	\end{cases}
\]
%W.l.o.g., we can assume that $c(e) \leq d(u,v)$, as otherwise we could modify $S$ by replacing $e$ with a shortest path between $u$ and $v$ in $G$ thus obtaining a better solution. 
For the sake of readability, in this section we drop the superscripts from $G'$ and $R'$. Furthermore, we denote by $c_{-e}$ the edge-cost function $c$ restricted to $G-e$ and by $d_{-e}$ the corresponding shortest-path metric.

Let $\xi = \epsilon/10$ and $h = 2^{2(1+\lceil 1/\xi \rceil )\lceil 2/\epsilon\rceil}$. The algorithm, whose pseudocode is reported in Algorithm~\ref{alg:E_plus}, first checks whether $e \in E(S)$. If this is not the case, then the algorithm returns $S$. If $e \in E(S)$, then let $S_u$ and $S_v$ denote the two trees of $S-e$ containing $u$ and $v$, respectively. The algorithm first computes a $k$-restricted Steiner tree $S_{u,\xi}$ (resp., $S_{v,\xi}$), with $k=2^{2+\lceil 1/\xi\rceil}$, of $S_u$ (resp., $S_v$) such that $u \in V(S_{u,\xi})$ (resp., $v \in V(S_{v,\xi})$). Then the algorithm computes a solution $S'$ via the call of $\guessold(I',S_{u,\xi}+S_{v,\xi},h)$, and, finally, it returns the cheapest solution between $S$ and $S'$. 
As we will see, the removal of $e$ from $S$ is necessary to guarantee that the cost of the processed Steiner forest of $I'$ (i.e., $S-e$) is upper bounded by the cost of $\opt'$.
\begin{algorithm}[h]
\footnotesize
\caption{The cost of an edge $e=(u,v)$ increases by $\Delta > 0$.}
	\label{alg:E_plus}

	\DontPrintSemicolon
	$\xi \gets \epsilon/10$;\;
	$h \gets 2^{2(1+\lceil 1/\xi \rceil )\lceil 2/\epsilon\rceil}$;\;
	\If{$e \not \in E(S)$}
	{
		$S' \gets S$;
	}
	\Else
	{
		let $S_u$ and $S_v$ be the tree of $S-e$ containing $u$ and $v$, respectively;\;
		$I_u \gets \langle G-e, c_{-e}, V(S_u)\cap R\rangle$;\;
		$I_v \gets \langle G-e, c_{-e}, V(S_v)\cap R\rangle$;\;
		$S_{u,\xi} \gets \borchersdu(I_u, S_u,\xi,u)$;\;
		$S_{v,\xi} \gets \borchersdu(I_v, S_v,\xi,v)$;\;
		$S_{\xi}\gets S_{u,\xi}+S_{v,\xi}$;\;
		$S' \gets \guessold(I',S_{\xi},h)$;\;
		{\bf if} $c'(S) \leq c'(S')$ {\bf then} $S' \gets S$;\;
	}
	\Return $S'$;\;
\end{algorithm}
\begin{theorem}
Let $\langle I, S, I' \rangle$ be an instance of Steiner tree reoptimization, where $S$ is a $\rho$-approximate solution of $I$ and $I'$ is obtained from $I$ by increasing the cost of a single edge. Then Algorithm~\ref{alg:E_plus} computes a $(\rho+\epsilon)$-approximate Steiner tree of $I'$ in polynomial time.
\end{theorem}
\begin{proof}
Corollary~\ref{cor:BorchersDu} implies that the computation of both $S_{u,\xi}$ and $S_{v,\xi}$ requires polynomial time. Therefore, using Lemma~\ref{lm:guessold} and the fact that $h$ is a constant value, the call of $\guessold(I',S_{\xi},h)$ outputs a solution $S'$ in polynomial time. Hence, the overall time complexity of Algorithm~\ref{alg:E_plus} is polynomial. 

In the following we show that Algorithm~\ref{alg:E_plus} returns a $(\rho+\epsilon)$-approximate solution. Observe that the local modification does not change the set of feasible solutions, i.e., a tree is a Steiner tree of $I$ iff it is a Steiner tree of $I'$. This implies that $\opt'$ is a Steiner tree of $I$ as well. Thanks to this observation, if $e \not\in E(S)$  we have that
\[
c'(S) = c(S) \leq \rho c(\opt) \leq \rho c(\opt') \leq \rho c'(\opt'),
\]
i.e., the solution returned by the algorithm is a $\rho$-approximate solution.
Moreover, if $e \in E(\opt')$, then $c'(\opt')=c(\opt')+\Delta$ or, equivalently, $c(\opt') = c'(\opt')-\Delta$. As a consequence, 
\[
c'(S) \leq c(S)+\Delta \leq \rho c(\opt) + \Delta \leq \rho c(\opt') + \Delta \leq \rho c'(\opt'),
\]
i.e., once again, the solution returned by the algorithm is a $\rho$-approximate solution.
Therefore, in the rest of the proof, we assume that $e \in E(S)$ as well as $e \not \in E(\opt')$. 

Let $S'_{\xi}$ denote the Steiner tree of $\langle K_n,d_{-e},R \rangle$ that is returned by the call of  $\borchersdu(\langle G-e,c_{-e},R\rangle, \opt',\xi)$. Using Theorem~\ref{thm:BorchersDu} and the fact that $e \not \in E(\opt')$ we have that 
\begin{equation}\label{eq:E_plus_extra}
d_{-e}(S'_{\xi}) \leq (1+\xi)c_{-e}(\opt') = (1+\xi)c'(\opt').
\end{equation}

\noindent Let $H'_0$ be a full component of $S'_{\xi}$ whose addition to $S_{\xi}$ yields a Steiner tree of $I'_n$ (ties are broken arbitrarily). For every $i=1,\dots,\ell$, with $\ell= \lceil 2/\epsilon \rceil$, we define:
\begin{itemize}
\item $H_i$ as the forest consisting of a minimal set of full components of $S_{\xi}$ such that the addition of $H_i$ and $e$ to $S'_{\xi}-H'_{i-1}$ yields a Steiner tree of $I_n$;\footnote{We observe that the existence of $H_i$ is guarantee by the fact that $S_{\xi}+e$ is a Steiner tree of $I_n$: in fact, $S_{\xi}=S_{u,\xi} + S_{v,\xi}$ is a forest of two trees such that $u \in S_{u,\xi}$ and $v \in S_{v,\xi}$.}

\item $H'_i$ as the forest consisting of a minimal set of full components of $S'_{\xi}$ whose addition to $S_{\xi}-H_{i}$ yields a Steiner tree of $I'_n$.
\end{itemize}

For the rest of the proof, we denote by $S_i$ the Steiner tree of $I'$ obtained by augmenting $S_{\xi}-H_{i}$ with $H'_i$. Observe that $S_i$ does not contain $e$. Therefore
\begin{equation}\label{eq:E_plus_zero}
c'(S_i) \leq c_{-e}(S_i) \leq d_{-e}(S_{\xi})-d_{-e}(H_{i})+d_{-e}(H'_i).
\end{equation}

Let $k=2^{2+\lceil 1/\xi \rceil }$ and $r=2^{\lceil 1/\xi \rceil}$. Observe that $S_{\xi}$ is a $k$-restricted Steiner forest of $I_n$ (see Corollary~\ref{cor:BorchersDu}), while $S'_{\xi}$ is an $r$-restricted Steiner tree of $I'_n$ (see Theorem~\ref{thm:BorchersDu}). Therefore, $H'_0$ spans at most $r$ terminals. Moreover, by repeatedly using Lemma~\ref{lm:full_connected_components}, $H_i$ contains at most $k^{i-1}r^{i}$ full components and spans at most $k^{i} r^{i}$ terminals, while $H'_i$ contains at most $k^{i}r^{i}$ full components and spans at most $k^{i} r^{i+1}$ terminals.
This implies that the number of full components of $H_i$, with $i=1,\dots,\ell$, is at most
\begin{equation*}
k^{i-1} r^{i} \leq k^{\ell} r^\ell=2^{2(1+\lceil 1/\xi \rceil )\lceil 2/\epsilon\rceil}= h.
\end{equation*}

Therefore, by Lemma~\ref{lm:merging_opts}, the cost of the solution returned by the call of $\guessold(I',S_{\xi},h)$ is at most the cost of $S_i$. As a consequence
\begin{equation}\label{eq:E_plus_sixth}
c'(S') \leq \frac{1}{\ell}\sum_{i=1}^\ell c'(S_i).
\end{equation}

Now we prove an upper bound to the cost of each $S_i$, with $i\in \{1,\dots,\ell\}$. Let $\Delta_{\opt}=(1+\xi)c'(\opt')-c(\opt)$ and observe that $\Delta_{\opt}\geq 0$.
Since the addition of $H_{i}$ and $e$ to $S'_{\xi}-H'_{i-1}$ yields a Steiner tree of $I$, the cost of this tree is lower bounded by the cost of $\opt$. As a consequence, using~(\ref{eq:E_plus_extra}) together with  Theorem~\ref{thm:BorchersDu} in the third inequality that follows, 
\begin{align*}
c(\opt)	 &\leq d(S'_{\xi}-H'_{i-1})+d(H_i)+d(e)
		  \leq d_{-e}(S'_{\xi}-H'_{i-1})+d_{-e}(H_i)+c(e)\\
		 &= d_{-e}(S'_{\xi}) - d_{-e}(H'_{i-1}) + d_{-e}(H_{i})+c(e) \\
		 &\leq (1+\xi)c'(\opt') - d_{-e}(H'_{i-1}) + d_{-e}(H_{i})+c(e),
\end{align*}
from which we derive
\begin{equation}\label{eq:E_plus_first}
d_{-e}(H_{i}) \geq  d_{-e}(H'_{i-1}) - \Delta_{\opt} - c(e). 
\end{equation}
Using Corollary~\ref{cor:BorchersDu} twice in the first inequality that follows we have that
\begin{align}\label{eq:E_plus_second}
d_{-e}(S_{\xi})	& = d_{-e}(S_{u,\xi}) +d_{-e}(S_{v,\xi}) \leq (1+\xi)c_{-e}(S_u)+(1+\xi)c_{-e}(S_v) \notag\\
				& = (1+\xi)c(S)-(1+\xi)c(e) \leq \rho(1+\xi)c(\opt)-c(e) \notag\\
				& \leq \rho(1+\xi)^2c'(\opt')-\Delta_\opt-c(e).
\end{align}
Starting from~(\ref{eq:E_plus_zero}) and using both~(\ref{eq:E_plus_first}) and~(\ref{eq:E_plus_second}) in the second inequality that follows, we can upper bound the cost of $S_i$, for every $i \in \{1,\dots,\ell\}$, with
\begin{equation}\label{eq:E_plus_fourth}
c'(S_i) 	\leq  d_{-e}(S_{\xi})-d_{-e}(H_{i})+d_{-e}(H'_i) \leq \rho(1+\xi)^2c'(\opt')-d_{-e}(H'_{i-1})+d_{-e}(H'_{i}).
\end{equation}

As a consequence, if for every $i\in\{1,\dots,\ell\}$ we substitute the term $c'(S_i)$ in~(\ref{eq:E_plus_sixth}) with the corresponding upper bound  in~(\ref{eq:E_plus_fourth}), and use~(\ref{eq:E_plus_extra}) to derive $d_{-e}(H'_\ell) \leq d_{-e}(S'_{\xi}) \leq (1+\xi)c'(\opt')$, we obtain
\begin{equation*}
c'(S')\leq \frac{1}{\ell}\sum_{i=1}^\ell c'(S_i) \leq \rho(1+\xi)^2 c'(\opt')+\frac{1+\xi}{\ell} c'(\opt')\leq (\rho+\epsilon)c'(\opt'),
\end{equation*}
where last inequality holds by the choice of $\xi$ and $\ell$, and because $\rho \leq 2$ and $\epsilon \leq 1$. This completes the proof.
\end{proof}

\section{A terminal becomes a Steiner vertex}\label{R_minus}

In this section we consider the local modification in which a terminal vertex $t \in R$ becomes a Steiner vertex. More precisely, the input instance satisfies the following: $G'=G$, $c'=c$, whereas $R'=R\setminus \{t\}$. Therefore, for the sake of readability, in this section we drop the superscripts from $c'$, $d'$, and $G'$. We recall that, unless $\rho=1$, the problem is as hard to approximate as the $\stp$~\cite{DBLP:conf/fsttcs/GoyalM15}. Therefore, in this section we assume that $S=\opt$ is an optimal solution, i.e., $\rho=1$.

Let $\xi = \epsilon/10$ and $h = (1+\lceil 1/\xi\rceil)2^{2(1+\lceil 1/\xi\rceil)\lceil 2/\epsilon\rceil}$.
As $c(\opt')$ may be much smaller than $c(\opt)$, we first need to remove a path from $\opt$ to ensure that the obtained Steiner forest of $I'$ has a constant number of trees and a cost upper bounded by $(1+\xi) c(\opt')$. The algorithm, whose pseudocode is reported in Algorithm~\ref{alg:R_minus}, precomputes a maximal (w.r.t. hop-length) path $P^*$ of $\opt$ having $t$ as one of its endvertices and containing up to $1+\lceil 1/\xi \rceil$ Steiner vertices whose corresponding degrees in $\opt$ are all greater than or equal to 3. As the following lemma shows, $P^*$ is exactly the path we are searching for.

\begin{algorithm}[t]
\footnotesize
\caption{A terminal $t$ of $G$ becomes a Steiner vertex.}
	\label{alg:R_minus}

	\DontPrintSemicolon
	$\xi \gets \epsilon/10$;\;
	$h \gets (1+\lceil 1/\xi\rceil)2^{2(1+\lceil 1/\xi\rceil)\lceil 2/\epsilon\rceil}$;\;
	{let $P^*$ be a maximal (w.r.t. hop-length) path of $\opt$ starting from $t$ and containing up to $1+\lceil 1/\xi \rceil$ Steiner vertices whose corresponding degrees in $\opt$ are all greater than or equal to 3};\;
	let $S_{\xi}$ be the empty graph;\;
	\For{every $v \in V(P^*)\cap V(\opt-P^*)$}
	{	
		let $S_v$ be the tree of $\opt-P^*$ that spans $v$;\;
		$S_{v,\xi} \gets \borchersdu(I,S_v,\xi,v)$;\;
		$S_{\xi} \gets S_{\xi} + S_{v,\xi}$;\;
	}
	$S' \gets \guessold(I',S_{\xi},h)$;\;
	{\bf if} $c(\opt) \leq c(S')$ {\bf then} \Return $\opt$ {\bf else} \Return $S'$;
\end{algorithm}

\begin{lemma}\label{lm:R_minus}
$\opt-P^*$ is a Steiner forest of $I'$ with at most $1+\lceil 1/\xi\rceil$ trees and $c(P^*) \geq c(\opt)-(1+\xi)c(\opt')$.
\end{lemma}
\begin{proof}
By definition of $P^*$, $\opt-P^*$ is a Steiner forest of $I'$ with at most $1+\lceil 1/\xi\rceil$ trees (we recall that only trees of the forest spanning one or more terminals are kept; therefore, Steiner vertices of $P^*$ having degree 2 in $S$ are not contained in $\opt-P^*$).

To show that $c(P^*) \geq c(\opt)-(1+\xi)c(\opt')$, we divide the proof into two complementary cases, depending on whether $P^*$ spans other terminals other than $t$.

We consider the case in which $P^*$ spans other terminals other than $t$. In this case $\opt'+P^*$ is a Steiner tree of $I$. Therefore, $c(\opt)\leq c(\opt')+c(P^*)$, from which we derive $c(P^*) \geq c(\opt)-c(\opt')\geq c(\opt)-(1+\xi)c(\opt')$.

Now, we consider the case in which $P^*$ does not span another terminal other than $t$. In this case $P^*$ contains exactly $1+\lceil 1/\xi\rceil$ Steiner vertices, each of which has degree greater than or equal to 3 in $\opt$. Therefore, $\opt-P^*$ yields a Steiner forest of $I'$ of exactly $1+\lceil 1/\xi\rceil$ trees. This implies that  $\opt-P^*$ contains at least $1+\lceil 1/\xi \rceil $ pairwise edge-disjoint paths connecting vertices of $P^*$ with terminals in $R'$. Let $P$ be any of these paths that minimizes its cost (ties are broken arbitrarily). As the addition of both $P^*$ and $P$ to $\opt'$ yields a Steiner tree of $I$, $(1+\lceil 1/\xi\rceil )c(P)+c(P^*)\leq c(\opt)\leq c(\opt')+c(P)+c(P^*)$, from which we derive $c(P)\leq \xi c(\opt')$. As a consequence $c(\opt)\leq c(\opt')+c(P)+c(P^*) \leq c(\opt')+\xi c(\opt')+c(P^*)$, i.e., $c(P^*) \geq c(\opt)-(1+\xi)c(\opt')$.
\end{proof}

The algorithm then computes a $k$-restricted Steiner forest $S_{\xi}$ of $\opt-P^*$, with $k=2^{2+\lceil 1/\xi \rceil}$, such that $S_{\xi}+P^*$ yields a Steiner tree of $I$. Finally, the algorithm calls $\guessold(I',S_{\xi},h)$ to compute a feasible solution $S'$ and returns the cheapest solution between $\opt$ and $S'$.

\begin{theorem}\label{thm:R_minus}
Let $\langle I, \opt, I' \rangle$ be an instance of Steiner tree reoptimization, where $\opt$ is an optimal solution of $I$ and $I'$ is obtained from $I$ by changing the status of a single vertex from terminal to Steiner. Then Algorithm~\ref{alg:R_minus} computes a $(1+\epsilon)$-approximate Steiner tree of $I'$ in polynomial time.
\end{theorem}
\begin{proof}
Corollary~\ref{cor:BorchersDu} implies that the computation of each $S_{v,\xi}$, with $v \in V(P^*)$, requires polynomial time. By Lemma~\ref{lm:R_minus}, $S_{\xi}$ is a Steiner forest of $I_n$ with a constant number of trees. Therefore, using Lemma~\ref{lm:guessold} and the fact that $h$ is a constant value, the call of $\guessold(I',S_{\xi},h)$ outputs a solution in polynomial time. Hence, the overall time complexity of Algorithm~\ref{alg:R_minus} is polynomial. 

In the following we show that Algorithm~\ref{alg:R_minus} returns a $(1+\epsilon)$-approximate solution. W.l.o.g., we can assume that $c(\opt)>(1+\epsilon)c(\opt')$ as otherwise the solution returned by the algorithm, whose cost is less than or equal to $c(\opt)$, would already be a $(1+\epsilon)$-approximate solution.

Let $S'_{\xi}$ be the Steiner tree of $I'_n$ that is returned by the call $\borchersdu(I',\opt',\xi)$ and let $H'_0$ be the forest consisting of a minimal set of full components of $S'_{\xi}$ whose addition to $S_{\xi}$ yields a Steiner tree of $I'_n$. For every $i=1,\dots,\ell$, with $\ell= \lceil 2/\epsilon \rceil$, we define:
\begin{itemize}
\item $H_i$ as the forest consisting of a minimal set of full components of $S_{\xi}$ such that the addition of $H_i$ and $P^*$ to $S'_{\xi}-H'_{i-1}$ yields a Steiner tree of $I_n$;\footnote{We observe that the existence of $H_i$ is guarantee by the fact that $S_{\xi}+P^*$ is a Steiner tree of $I_n$.}
\item $H'_i$ as the forest consisting of a minimal set of full components of $S'_{\xi}$ whose addition to $S_{\xi}-H_{i}$ yields a Steiner tree of $I'_n$.
\end{itemize}

For the rest of the proof, we denote by $S_i$ the Steiner tree of $I'$ obtained by augmenting $S_{\xi}-H_{i}$ with $H'_i$.

Let $k=2^{2+\lceil 1/\xi \rceil}$ and $r=2^{\lceil 1/\xi \rceil}$. Observe that $S_{\xi}$ is a $k$-restricted Steiner forest of $I_n$ (see Corollary~\ref{cor:BorchersDu}), while $S'_{\xi}$ is an $r$-restricted Steiner tree of $I'_n$ (see Theorem~\ref{thm:BorchersDu}). Using Lemma~\ref{lm:full_connected_components}, $H'_0$ contains at most $(1+\lceil 1/\xi\rceil)$ full components and spans at most $(1+\lceil 1/\xi\rceil) r$ terminals. By repeatedly using Lemma~\ref{lm:full_connected_components}, $H_i$ contains at most $(1+\lceil 1/\xi\rceil) k^{i-1}r^{i}$ full components and spans at most $(1+\lceil 1/\xi\rceil) k^{i}r^{i}$ terminals, while $H'_i$ contains at most $(1+\lceil 1/\xi\rceil) k^{i}r^{i}$ full components and spans at most $(1+\lceil 1/\xi\rceil) k^{i}r^{i+1}$ terminals. 
This implies that the number of full components of $H_i$, with $i=1,\dots,\ell$, is at most
\begin{equation*}
(1+\lceil 1/\xi\rceil) k^{i-1}r^{i} \leq (1+\lceil 1/\xi\rceil) k^{\ell}r^{\ell}= (1+\lceil 1/\xi\rceil)2^{2(1+\lceil 1/\xi\rceil)\lceil 2/\epsilon\rceil}= h.
\end{equation*}
Therefore, by Lemma~\ref{lm:merging_opts}, the cost of the solution returned by the algorithm is at most the cost of $S_i$, for every $i=1,\dots,\ell$. As a consequence
\begin{equation}\label{eq:R_minus_sixth}
c(S') \leq \frac{1}{\ell}\sum_{i=1}^\ell c(S_i).
\end{equation}

Now we prove an upper bound to the cost of each $S_i$.
Let $\Delta_{\opt}=c(\opt)-(1+\xi)c(\opt')$ and observe that $\Delta_{\opt} \geq 0$ by our assumption.
Since the addition of $H_{i}$ and $P^*$ to $S'_{\xi}-H'_{i-1}$ yields a Steiner tree of $I$, the cost of this tree is lower bounded by the cost of $\opt$. As a consequence, using Theorem~\ref{thm:BorchersDu} in the second inequality that follows, 
\begin{equation*}
c(\opt)	\leq d(S'_{\xi})-d(H'_{i-1})+d(H_{i})+d(P^*)\leq (1+\xi)c(\opt') - d(H'_{i-1}) + d(H_{i})+c(P^*),
\end{equation*}
from which we derive
\begin{equation}\label{eq:R_minus_first}
d(H_{i}) \geq  d(H'_{i-1}) + \Delta_{\opt} - c(P^*). 
\end{equation}
Using Corollary~\ref{cor:BorchersDu} and Lemma~\ref{lm:R_minus} in the first  and last inequality that follows, respectively, we have that
\begin{align}\label{eq:R_minus_second}
d(S_{\xi})	& = \sum_{v \in V(P^*)}d(S_{v,\xi}) \leq (1+\xi)c(\opt-P^*) \leq (1+\xi)c(\opt)-(1+\xi)c(P^*)\notag\\
			& \leq (1+\xi)^2c(\opt')+(1+\xi)\Delta_\opt-(1+\xi)c(P^*)\notag\\
			& \leq (1+\xi)^2c(\opt')+\Delta_\opt-c(P^*).
\end{align}
Using both~(\ref{eq:R_minus_first}) and~(\ref{eq:R_minus_second}) in the second inequality that follows, we can upper bound the cost of $S_i$, for every $i \in \{1,\dots,\ell\}$, with
\begin{equation}\label{eq:R_minus_fourth}
c(S_i)	\leq  d(S_{\xi})-d(H_{i})+d(H'_i) \leq (1+\xi)^2c(\opt')-d(H'_{i-1})+d(H'_{i}).
\end{equation}

Therefore, if for every $i\in\{1,\dots,\ell\}$ we substitute the term $c(S_i)$ in~(\ref{eq:R_minus_sixth}) with the corresponding upper bound  in~(\ref{eq:R_minus_fourth}), and use Theorem~\ref{thm:BorchersDu} to derive $d(H'_\ell) \leq d(S'_{\xi}) \leq (1+\xi)c(\opt')$, we obtain
\begin{equation*}
c(S')\leq \frac{1}{\ell}\sum_{i=1}^\ell c(S_i) \leq (1+\xi)^2 c(\opt')+\frac{1+\xi}{\ell} c(\opt')\leq (1+\epsilon)c(\opt'),
\end{equation*}
where last inequality holds by the choice of $\xi$ and $\ell$, and because $\epsilon \leq 1$. This completes the proof.
\end{proof}

\section{The cost of an edge decreases}\label{E_minus}

In this section we consider the local modification in which the cost of an edge $e=(u,v) \in E$ decreases by $0<\Delta \leq c(e)$. More precisely, the input instance satisfies the following: $G'=G$, $R'=R$, whereas for every $e' \in E$,
\[
c'(e')=
	\begin{cases}
		c(e) - \Delta	&	\text{if $e'=e$;}\\
		c(e')			&	\text{otherwise.}
	\end{cases}
\]

For the sake of readability, in this section we drop the superscripts from $G'$ and $R'$. Furthermore, for the rest of this section, we denote by $c_{-e}$ the edge-cost function $c$ restricted to $G-e$ and by $d_{-e}$ the corresponding shortest-path metric. We recall that, unless $\rho=1$, the problem is as hard to approximate as the $\stp$~\cite{DBLP:conf/fsttcs/GoyalM15}. Therefore, in this section we assume that $S=\opt$ is an optimal solution, i.e., $\rho=1$.

Let $\xi = \epsilon/10$ and $h = 2(1+\lceil 1/\xi\rceil)2^{2(2+\lceil 1/\xi\rceil)\lceil 2/\epsilon\rceil}$.
As $c'(\opt')$ may be much smaller than $c'(\opt)$, we first need to remove a path from $\opt$ to ensure that the obtained Steiner forest of $I'$ has a constant number of trees and a cost upper bounded by $(1+\xi) c'(\opt')$.  The algorithm, whose pseudocode is reported in Algorithm~\ref{alg:E_minus}, precomputes a set ${\cal H}$ containing all the paths of $\opt$ with up to $2(1+\lceil 1/\xi \rceil)$ Steiner vertices whose corresponding degrees in $\opt$ are all greater than or equal to 3. As the following lemma shows, there exists $P^* \in {\cal H}$, which is exactly the path we are searching for when $c'(\opt')$ is much smaller than $c'(\opt)$.

\begin{algorithm}[t]
\footnotesize
\caption{The cost of an edge $e=(u,v)$ of $G$ decreases by $0 < \Delta \leq c(e)$.}
	\label{alg:E_minus}

	\DontPrintSemicolon
	$\xi \gets \epsilon/10$;\;
	$h \gets 2(1+\lceil 1/\xi\rceil)2^{2(2+\lceil 1/\xi\rceil)\lceil 2/\epsilon\rceil}$;\;
	$S' \gets \opt$;\;
	{let ${\cal H}$ be the set of all the paths of $\opt$, each of which contains up to $2(1+\lceil 1/\xi \rceil)$ Steiner vertices whose corresponding degrees in $\opt$ are all greater than or equal to 3};\;
	
	\For{every $P \in {\cal H}$}
	{	
		let $S_{\xi}$ be the empty graph;\;
		\For{every $v \in V(P)\cap V(\opt - P)$}
		{	
			let $S_v$ be the tree of $F$ that spans $v$;\;
			$S_{v,\xi} \gets \borchersdu(\langle G-e,c_{-e},R\rangle,S_v,\xi,v)$;\;
			$S_{\xi} \gets S_{\xi} + S_{v,\xi}$;\;
		}
		$S'' \gets \guessold(I',S_{\xi},h)$;\;
		{\bf if} $c'(S'') \leq c'(S')$ {\bf then} $S'\gets S''$;\;
	}
	\Return $S'$;

\end{algorithm}

\begin{lemma}\label{lm:E_minus}
If $e \in E(\opt')$ and $e \not \in E(\opt)$, then there exists $P^* \in {\cal H}$ s.t. $c'(P^*) \geq c(\opt)-(1+\xi)c'(\opt')$.
\end{lemma}
\begin{proof}
Let $S'_u$ and $S'_v$ denote the trees of $\opt'-e$ containing $u$ and $v$, respectively. Let $R_u=V(S'_u)\cap R$ and $R_v=V(S'_v)\cap R$. As $e \not \in E(\opt)$, there exists a path in $\opt$, say $\bar P$, that connects a terminal in $R_u$, say $u'$, with a terminal in $R_v$, say $v'$, and such that all the internal vertices of $\bar P$ are Steiner vertices. Let $\sigma$ be the number of Steiner vertices of $\bar P$ whose corresponding degrees in $\opt$ are all greater than or equal to 3.
We divide the proof into two complementary cases depending on the value of $\sigma$. In both cases, we prove that $P^*$ is a subpath of $\bar P$.

We consider the case in which $\sigma \leq 2(1+\lceil 1/\xi\rceil)$. Let $P^*=\bar P$. Clearly $P^* \in {\cal H}$ and $\opt'-e+P^*$ is a Steiner tree of $I$. Therefore, $c(\opt)\leq c'(\opt'-e)+c'(P^*)\leq c'(\opt')+c'(P^*)$, from which we derive $c'(P^*) \geq c(\opt)-c'(\opt')\geq c(\opt)-(1+\xi)c'(\opt')$.

Now, we consider the case in which $\sigma >  2(1+\lceil 1/\xi\rceil)$. Let $X$ be the set of the $\sigma$ Steiner vertices of $\bar P$ whose corresponding degrees in $\opt$ are all greater than or equal to 3. Let $x_1,\dots,x_{\sigma}$ be the vertices in $X$ in order in which they are encountered while traversing $\bar P$ from $u'$ to $v'$. Finally, let $x_0=u'$ and $x_{\sigma+1}=v'$. We associate a label $\lambda(i)\in\{u,v\}$ with each $x_i$ to indicate the existence of a path in $\opt$ between $x_i$ and a terminal in $R_{\lambda(i)}$ that is edge-disjoint from $\bar P$ (if $\lambda(i)$ can assume both values, we arbitrarily choose one of them). We set $\lambda(0)=u$ and $\lambda(\sigma+1)=v$ to indicate that $u'$ (resp., $v'$) has the trivial paths towards itself. For $i=0,\dots, \sigma-2\lceil 1/\xi\rceil$, let $L_i=\lambda(i),\dots,\lambda(i+1+2\lceil 1/\xi\rceil))$ be the sequence of $2(1+\lceil 1/x\rceil)$ consecutive labels starting from $i$. One of the following holds:
\begin{description}
\item[(a)] $L_0$ contains at least $1+\lceil 1/\xi\rceil$ occurrences of $v$;
\item[(b)] $L_{\sigma-2\lceil 1/\xi\rceil}$ contains at least $1+\lceil 1/\xi\rceil$ occurrences of $v$;
\item[(c)] there is an index $i$ such that the sequence $L_i$ contains exactly $1+\lceil 1/\xi\rceil$ occurrences of $u$ and $1+\lceil 1/\xi\rceil$ occurrences of $v$.
\end{description}
Indeed, let $\mu_i$ and $\nu_i$ denote the occurrences of $u$ and $v$ in $L_i$, respectively. Observe that $|\mu_{i}-\mu_{i+1}|\in\{0,1\}$ as well as $|\nu_{i}-\nu_{i+1}|\in\{0,1\}$. More precisely:
\begin{itemize}
\item $\mu_i=\mu_{i+1}$ iff $\nu_i=\nu_{i+i}$;
\item $\mu_{i}=\mu_{i+1}+1$ iff $\nu_{i}=\nu_{i+1}-1$;
\item $\mu_{i}=\mu_{i+1}-1$ iff $\nu_{i}=\nu_{i+1}+1$.
\end{itemize}
Therefore, if neither (a) nor (b) holds, i.e., $\mu_0 > \nu_0$ and $\mu_{\sigma-2\lceil 1/\xi\rceil} < \nu_{\sigma-2\lceil 1/\xi\rceil}$, then (c) must hold. 

Let $P^*$ be the path containing exactly the Steiner vertices corresponding to a sequence $L_i$ that satisfies any of the above three conditions (ties are broken arbitrarily). Observe that $P^* \in {\cal H}$. Let $P_u$ (resp., $P_v$) be a minimum-cost path between a vertex of $P^*$ and a vertex of $R_u$ (resp., $R_v$) that is edge-disjoint from $P^*$. W.l.o.g., if $P^*$ contains $u'$, then we can assume that $P_u$ is the path containing only vertex $u$; similarly, if $P^*$ contains $v'$, then we can assume that $P_v$ is the path containing only vertex $v$. As the addition of both $P^*$, $P_u$ and $P_v$ to $\opt'-e$ yields a Steiner tree of $I$, $(1+\lceil 1/\xi\rceil )(c'(P_u)+c'(P_v))+c'(P^*)=(1+\lceil 1/\xi\rceil )(c(P_u)+c(P_v))+c(P^*)\leq c(\opt)\leq c'(\opt')+c'(P_u)+c'(P_v)+c'(P^*)$, from which we derive $c'(P_u)+c'(P_v)\leq \xi c'(\opt')$. As a consequence $c(\opt)\leq c'(\opt')+c'(P_u)+c'(P_v)+c'(P^*) \leq c'(\opt')+\xi c'(\opt')+c'(P^*)$, i.e., $c'(P^*) \geq c(\opt)-(1+\xi)c(\opt')$.
\end{proof}

The algorithm always returns a solution whose cost is upper bounded by $c'(\opt)$. This allows us to guarantee that the output of the algorithm is a $(1+\epsilon)$-approximate solution when $e \not \in E(\opt')$ or $e \in E(\opt)$. In the complementary case in which $e \in E(\opt')$ and $e \not \in E(\opt)$ the algorithm computes a $(1+\epsilon)$-approximate solution using the path $P^*$ that satisfies the conditions stated in Lemma~\ref{lm:E_minus}. More precisely, the algorithm computes a $k$-restricted Steiner forest $S_{\xi}$ of $\opt-P^*$ w.r.t. the $\stp$ instance $\langle G-e,c_{-e},R\rangle$, with $k=2^{2+\lceil 1/\xi \rceil}$, such that $S_{\xi}+P^*$ yields a Steiner tree of $I$. Finally, the algorithm calls $\guessold(I',S_{\xi},h)$ to compute a feasible solution $S''$ and returns a solution whose cost is upper bounded by $c'(S'')$.

\begin{theorem}
Let $\langle I, \opt, I' \rangle$ be an instance of Steiner tree reoptimization, where $\opt$ is an optimal solution of $I$ and $I'$ is obtained from $I$ by decreasing the cost of a single edge. Then Algorithm~\ref{alg:E_minus} computes a $(1+\epsilon)$-approximate Steiner tree of $I'$ in polynomial time.
\end{theorem}
\begin{proof}
Corollary~\ref{cor:BorchersDu} implies that the computation of each $S_{v,\xi}$, with $v \in V(P)$, requires polynomial time. Observe that $\opt-P$ is a Steiner forest of $I$ of at most $2(1+\lceil 1/\xi \rceil)$ trees. As a consequence $S_{\xi}$ is a Steiner forest of $I_n$ with at most $2(1+\lceil 1/\xi \rceil)$ trees. Therefore, using Lemma~\ref{lm:guessold} and the fact that $h$ is a constant value, the call of $\guessold(I',S_{\xi},h)$ outputs a solution in polynomial time. Finally, since ${\cal H}$ contains at most $O(n^2)$ paths (one path for each pair of vertices), the overall time complexity of Algorithm~\ref{alg:E_minus} is polynomial. 

In the following we show that Algorithm~\ref{alg:E_minus} returns a $(1+\epsilon)$-approximate solution. W.l.o.g., we can assume that $c(\opt)>(1+\epsilon)c(\opt')$, as otherwise, the solution returned by the algorithm, whose cost is always less then or equal to $c'(\opt)\leq c(\opt)$, would already be a $(1+\epsilon)$-approximate solution. As a consequence, we have that $e \not \in E(\opt)$  as well as $e \in E(\opt')$: indeed, if even one of these two conditions were not satisfied, then $\opt$ would be an optimal solution of $I'$.

Let $S'_u$ and $S'_v$ be the two trees of $\opt'-e$ that contain $u$ and $v$, respectively. Let $S'_{u,\xi}$ and $S'_{v,\xi}$ be the two $k$-restricted trees of $I'_n$, with $k=2^{2+1/\xi}$, that are  returned by the call $\borchersdu(\langle G-e,c_{-e},R\rangle,S'_u,\xi,u)$ and $\borchersdu(\langle G-e,c_{-e},R\rangle,S'_v,\xi,v)$, respectively. Finally, let $S'_{\xi}=S'_{u,\xi}+S'_{v,\xi}$, $P^* \in {\cal H}$ the path satisfying the conditions stated in Lemma~\ref{lm:E_minus}, and $S_{\xi}$ the $k$-restricted Steiner forest of $I$ that is computed by the algorithm when $P=P^*$. Let $H'_0$ be the forest consisting of a minimal set of full components of  $S'_{\xi}$ such that the addition of $H'_0$ and $e$ to $S_{\xi}+e$ yields a Steiner tree of $I'_n$.
For every $i=1,\dots,\ell$, with $\ell= \lceil 2/\epsilon \rceil$, we define:
\begin{itemize}
\item $H_i$ as the forest consisting of a minimal set of full components of $S_{\xi}$ such that the addition of $H_i$ and $P^*$ to $S'_{\xi}-H'_{i-1}$ yields a Steiner tree of $I_n$;\footnote{We observe that the existence of $H_i$ is guarantee by the fact that $S_{\xi}+P^*$ is a Steiner tree of $I_n$.}
\item $H'_i$ as the forest consisting of a minimal set of full components of $S'_{\xi}$ such that the addition of $H'_i$ and $e$ to $S_{\xi}-H_{i}$ yields a Steiner tree of $I'_n$.\footnote{We observe that the existence of $H_i$ is guarantee by the fact that $S'_{\xi}+e$ is a Steiner tree of $I'_n$.}

\end{itemize}

For the rest of the proof, we denote by $S_i$ the Steiner tree of $I'$ obtained by augmenting $S_{\xi}-H_{i}$ with $H'_i$ and $e$.

Let $k=2^{2+\lceil 1/\xi \rceil}$. Observe that $S_{\xi}$ and $S'_{\xi}$ are $k$-restricted Steiner forests of $I_n$ and $I'_n$, respectively (see Corollary~\ref{cor:BorchersDu}). Observe that $S_{\xi}+e$ is a forest of at most $1+2(1+\lceil 1/\xi\rceil)$ trees. Therefore, using Lemma~\ref{lm:full_connected_components}, $H'_0$ contains at most $2(1+\lceil 1/\xi\rceil)$ full components and spans at most $2(1+\lceil 1/\xi\rceil)k$ terminals. Therefore, by repeatedly using Lemma~\ref{lm:full_connected_components}, $H_i$ contains at most $2(1+\lceil 1/\xi\rceil) k^{2i-1}$ full components and spans at most $2(1+\lceil 1/\xi\rceil) k^{2i}$ terminals, while $H'_i$ contains at most $2(1+\lceil 1/\xi\rceil) k^{2i)}$ full components and spans at most $2(1+\lceil 1/\xi\rceil) k^{2i+1}$ terminals. 
This implies that the number of full components of $H_i$, with $i=1,\dots,\ell$, is at most
\begin{equation*}
2(1+\lceil 1/\xi\rceil) k^{2i-1} \leq 2(1+\lceil 1/\xi\rceil) k^{2\ell}= 2(1+\lceil 1/\xi\rceil)2^{2(2+\lceil 1/\xi\rceil)\lceil 2/\epsilon\rceil}= h.
\end{equation*}
Therefore, by Lemma~\ref{lm:merging_opts}, the cost of the solution returned by the algorithm cannot be worse than the cost of $S_i$, for every $i=1,\dots,\ell$. As a consequence
\begin{equation}\label{eq:E_minus_sixth}
c'(S') \leq \frac{1}{\ell}\sum_{i=1}^\ell c'(S_i).
\end{equation}

Now we prove an upper bound to the cost of each $S_i$.
Let $\Delta_{\opt}=c(\opt)-(1+\xi)c'(\opt')$ and observe that $\Delta_{\opt} \geq  0$ by our assumption.
Since the addition of $H_{i}$ and $P^*$ to $S'_{\xi}-H'_{i-1}$ yields a Steiner tree of $I$, the cost of this tree is lower bounded by the cost of $\opt$. As a consequence, using Corollary~\ref{cor:BorchersDu} in the third inequality that follows, 
\begin{align*}
c(\opt)	& \leq d_{-e}(S'_{\xi})-d_{-e}(H'_{i-1})+d_{-e}(H_{i})+d_{-e}(P^*) \\
		& \leq (1+\xi)c_{-e}(\opt'-e) - d_{-e}(H'_{i-1}) + d_{-e}(H_{i})+d_{-e}(P^*) \\
		& \leq (1+\xi)c'(\opt')-c'(e) - d_{-e}(H'_{i-1}) + d_{-e}(H_{i})+d_{-e}(P^*) \\
		& \leq (1+\xi)c'(\opt') - d_{-e}(H'_{i-1}) + d_{-e}(H_{i})+c'(P^*)-c'(e),
\end{align*}
from which we derive
\begin{equation}\label{eq:E_minus_first}
d_{-e}(H_{i}) \geq  d_{-e}(H'_{i-1}) + \Delta_{\opt} - c'(P^*)+c'(e). 
\end{equation}
Using Corollary~\ref{cor:BorchersDu} and  Lemma~\ref{lm:E_minus} respectively in the first and last inequality that follows we have that
\begin{align}\label{eq:E_minus_second}
d_{-e}(S_{\xi})	& = \sum_{v \in V(P^*)}d_{-e}(S_{v,\xi}) \leq (1+\xi)c_{-e}(\opt-P^*)\notag \\
				& \leq (1+\xi)c(\opt)-(1+\xi)c(P^*)\notag\\
				& \leq (1+\xi)^2c'(\opt')+(1+\xi)\Delta_\opt-(1+\xi)c'(P^*)\notag \\
				& \leq (1+\xi)^2c'(\opt')+\Delta_\opt-c'(P^*).
\end{align}
Using both~(\ref{eq:E_minus_first}) and~(\ref{eq:E_minus_second}) in the second inequality that follows, we can upper bound the cost of $S_i$, for every $i \in \{1,\dots,\ell\}$, with
\begin{align}\label{eq:E_minus_fourth}
c'(S_i)	& \leq  d_{-e}(S_{\xi})-d_{-e}(H_{i})+d_{-e}(H'_i)+c'(e)\notag \\ 
		& \leq (1+\xi)^2c'(\opt')-d_{-e}(H'_{i-1})+d_{-e}(H'_{i}).
\end{align}

Therefore, if for every $i\in\{1,\dots,\ell\}$ we substitute the term $c(S_i)$ in~(\ref{eq:E_minus_sixth}) with the corresponding upper bound in~(\ref{eq:E_minus_fourth}), and use Corollary~\ref{cor:BorchersDu} to derive $d_{-e}(H'_\ell) +c'(e)\leq (1+\xi)c_{-e}(\opt'-e)+c'(e)\leq (1+\xi)c'(\opt')$, we obtain
\begin{equation*}
c'(S')\leq \frac{1}{\ell}\sum_{i=1}^\ell c'(S_i) \leq (1+\xi)^2 c'(\opt')+\frac{1+\xi}{\ell} c'(\opt')\leq (1+\epsilon)c'(\opt'),
\end{equation*}
where last inequality holds by the choice of $\xi$ and $\ell$, and because $\epsilon \leq 1$. This completes the proof.
\end{proof}

%\section{Conclusions}\label{conclusions}

%\todo[inline]{Remark that our results are interesting from a theoretical point-of-view as the time complexities of our algorithms, albeit polynomial in $n$, is infeasible to be used in the practise. Remark also that most of the previous algorithms either are large polynomial in $n$ or contain huge constant factor, both based from the fact that an extensive usage of brute force search algorithms for relatively-small instances are used. Therefore, an interesting scenario is that of proposing approximation algorithms which guarantee a worse approximation ratio, but runs in small polynomial time.}

\bibliography{bibliography}

%\newpage

%\appendix

\end{document}